\documentclass[letterpaper11pt]{article}

\usepackage{dynkin-diagrams}
\usepackage{cite}
\usepackage{amsmath}
\usepackage{amssymb}
\usepackage{amsthm}
\usepackage{mathrsfs}
\usepackage{xypic}
\usepackage{tikz-cd}
\usepackage{makecell}
\usepackage{afterpage}

\usepackage{hyperref}
\usepackage{color}
\definecolor{darkred}{rgb}{0.65,0.15,0}
\hypersetup{pdfborder={0 0 0},colorlinks=true,urlcolor=darkred,citecolor=blue,linkcolor=darkred,linktocpage=true}

\usepackage{float} 

\usepackage{geometry}

\newgeometry{vmargin={35mm}, hmargin={35mm}} 

\usepackage{keyval}
\makeatletter
\define@key{setpar}{left}[0pt]{\leftmargin=#1}
\define@key{setpar}{right}[0pt]{\rightmargin=#1}
\define@key{setpar}{both}{\leftmargin=#1\relax\rightmargin=#1}
\makeatother

\makeatletter
\renewcommand*\env@matrix[1][\arraystretch]{%
  \edef\arraystretch{#1}%
  \hskip -\arraycolsep
  \let\@ifnextchar\new@ifnextchar
  \array{*\c@MaxMatrixCols c}}
\makeatother  
  
\def\ie{{\it i.e.}}

\def\EWeight#1#2#3#4{\bigl({}^{\mathstrut}_{#1\mathstrut}{}_{#2\mathstrut}^{#4\mathstrut}{}_{#3\mathstrut}^{\mathstrut}\bigr)}

\def\ESixWeight#1#2#3#4#5#6{\EWeight{#1#2}#3{#4#5}#6}
\def\ESix#1#2#3#4#5#6{$\ESixWeight#1#3#4#5#6#2$}

\def\ESeven#1#2#3#4#5#6#7{$\EWeight{#7#6#5}#4{#3#1}#2$}

\def\EEight#1#2#3#4#5#6#7#8{$\EWeight{#8#7#6#5}#4{#3#1}#2$}

\def\DWeight#1#2#3{\bigl(\raise2.5pt\hbox{${}_{#1}$}{}^{#2}_{#3}\bigr)}

\def\AAWeight#1#2{\bigl(\raise0pt\hbox{${}^{#1}_{#2}$}\bigr)}

\def\GTwo#1#2{(#2#1)}

\def\fg{{\mathfrak g}}

\def\so{{\mathfrak{so}}}
\def\su{{\mathfrak{su}}}
\def\sp{{\mathfrak{sp}}}
\def\sl{{\mathfrak{sl}}}

\def\nn{\nonumber}

\def\*{\partial}

\def\Red#1{\textcolor{red}{#1}}

\def\rank{\hbox{rank}}

   \def\tr{\hbox{tr}}

\newtheorem{prop}{Proposition}
\newtheorem{lemma}{Lemma}
\newtheorem{cor}{Corollary}
\newtheorem{thm}{Theorem}

\begin{document}

\frenchspacing

\includegraphics[height=2cm]{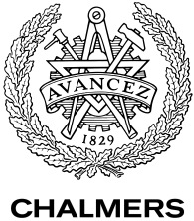}

\vspace{-14mm}
{\flushright Gothenburg preprint \\ 
November, 2023 \\ 
}

\vspace{4mm}

\hrule

\vspace{16mm}

\thispagestyle{empty}
\begin{center}
{\Large\bf\sc The chiral ring of $D=4$, ${\mathcal N}=1$ SYM\\[10pt]
with exceptional gauge groups}
   \\[14mm]
    
{\large
Martin Cederwall and Gabriele Ferretti}

\vspace{14mm}
 {
 {\it Department of Physics,
         Chalmers Univ. of Technology,\\
SE-412 96 Gothenburg, Sweden}}

\end{center}

\vfill

\begin{quote}
\textbf{Abstract:} 
The Cachazo--Douglas--Seiberg--Witten conjecture\cite{Cachazo:2002ry}, concerning the algebraic structure of the chiral ring in ${\mathcal N}=1$, $D=4$ supersymmetric Yang--Mills theory, is proven for exceptional gauge groups. 
This completes the proof of the conjecture. 
\end{quote} 

\vfill

\hrule

\noindent{\tiny email:
martin.cederwall@chalmers.se, gabriele.ferretti@chalmers.se}

\newpage


\noindent{\large\bf 1. Introduction}

\noindent Four-dimensional ${\mathcal{N}}=1$ supersymmetric pure Yang-Mills theories display a rich variety of phenomena of interest to physics. For any finite-dimensional, compact, simple Lie algebra\footnote{The global structure of the associated group will play no role in this work.} $\fg$, the dynamical fields consist of a gluon $A_\mu^a$ and a gluino $\lambda_\alpha^a$ combined into an anticommuting chiral superfield $W_\alpha^a$, where $\alpha=1,2$ is the $\sl(2)$ index and $a$ the adjoint index of $\fg$.

The theory is asymptotically free and possesses a discrete chiral symmetry ${\mathbb{Z}}_{2g^\vee}$, $g^\vee $ being the dual Coxeter number of $\fg$. This symmetry is the left-over of the axial $U(1)$ symmetry not explicitly broken by the ABJ anomaly.
The theory is also strongly believed to be confining and gapped with a number of supersymmetric vacua equal to $g^\vee $. 

It is of interest to consider the particularly well behaved ring of gauge invariant chiral superfields, consisting of all the gauge invariant local operators constructed out of  $W_\alpha^a$ (but not  $W^{\dagger a}_{\dot\alpha}$), modulo identification by operators containing the terms $f_{bc}{}^aW_\alpha^b W_\beta^c$ for any $\alpha$, $\beta$ and $a$, with $f_{bc}{}^a$ the structure constants of $\fg$.
The identification is needed because such terms are descendants in the chiral ring, \ie, can be written as $\{\bar Q_{\dot\alpha}, \Phi^{\dot\alpha}\}$ for some gauge invariant operator $\Phi^{\dot\alpha}$ and thus do not contribute to the correlation functions.

The most obvious element of the chiral ring is the invariant $S=\frac{1}{2}W^{\alpha a}W_{\alpha a}\equiv W^a_1W_{2a}$, which is non-zero by virtue of the anticommuting nature of $W_\alpha^a$. The overall normalization of $S$ is not relevant for this work. The classical chiral ring is constructed simply by treating $W^a_1$ and $W_2^a$ as classical Grassmann numbers.

The conjecture put forward in \cite{Cachazo:2002ry} by Cachazo, Douglas, Seiberg and Witten (hereafter CDSW), is that the classical chiral ring is generated by $S$ alone, together with the relation $S^{g^\vee}=0$ (but $S^{g^\vee-1}\neq 0$). This is important to physics because instanton corrections would then deform the classical chiral ring to $S^{g^\vee}=\Lambda^{3 g^\vee}$, where $\Lambda$ is the non-perturbative (holomorphic) scale. Factorization of chiral operators then implies that $S$ itself acquires a vacuum expectation value, breaking the ${\mathbb{Z}}_{2g^\vee}$ discrete chiral symmetry to ${\mathbb{Z}}_2$, thus giving rise to the $g^\vee$ vacua mentioned above.

The CDSW conjecture is a precisely formulated statement about Lie algebras that can be investigated independently of the currently incomplete understanding of the dynamics of quantum gauge theories. It has been proven for $\fg=A_n\equiv \su(n+1)$ already  in \cite{Cachazo:2002ry},  for the remaining classical Lie algebras  $\fg=B_n, C_n, D_n \equiv \so(2n+1), \sp(2n), \so(2n)$ in \cite{Witten:2003ye} and for $\fg=G_2$ in \cite{EtingofKac}. The first part of the conjecture, namely that the ring is generated by $S$ alone, has been shown for all $\fg$ in \cite{Kumar2004}. 

In this note we show that the CDSW conjecture holds for all the exceptional Lie algebras  $\fg=G_2, F_4, E_6, E_7, E_8$, thus completing the proof. (For completeness, this also includes the previous results of \cite{EtingofKac} and \cite{Kumar2004}.) 

\vspace{2\parskip}
\noindent{\large\bf 2. The proof}

\noindent For a given $\fg$, $\wedge\fg$ is a graded ring, generated by odd elements $X=X^aT_a\in\fg$ at level 1. 
For ease of notation, we let $W_1=X$, $W_2=Y$.
Let the graded ring $B$ be defined as $\wedge\fg/\langle[X,X]\rangle$, where the ideal is generated by 
$[X,X]=P^\fg X^2=-{1\over2g^\vee}f_{ab}{}^ef^{cd}{}_eX^aX^bT_c\wedge T_d$.
When we write $X^n$, this means ``wedge product'',
$X^n=X^{a_1}\ldots X^{a_n}T_{a_1}\wedge\ldots\wedge T_{a_n}$, which can then be projected on irreducible modules at a given level.

The ring $B$ has been studied and has many interesting properties \cite{Kostant2000}. Its decomposition into irreducible $\fg$-modules $\{r_i\}$ contains $2^r$ distinct (\ie, different) modules
\cite{PetersonUnpub}, where $r=\rank(\fg)$. 
They are in close correspondence to Abelian ideals in a Borel subalgebra of $\fg$ \cite{Kostant1998}.
In addition, the modules appearing at level $n$ (\ie, in $\wedge^n\fg$) are precisely those having a value of the quadratic Casimir operator $C_2$ which is $n$ times that of the adjoint. 
In Appendix \ref{ModulesApp}, we list the $r_i$'s by level for exceptional $\fg$. 
The method to compute them we found most efficient was the sequential use of selection by value of $C_2$.
The program LiE \cite{LiE} has been helpful in the calculations.

We then consider the graded ring $B\times B$, generated by $X$ and $Y$. In it, we will eventually divide out the ideal generated by $[X,Y]$, and denote the result $A=(B\times B)/\langle[X,Y]\rangle$. In particular, we will verify the 
CDSW conjecture \cite{Cachazo:2002ry} that the subring $A^\fg$ of $\fg$-invariants, the so-called classical chiral ring,
is generated by the scalar product $S=C_2(X,Y)$ (evaluated in some representation), and that furthermore
$S^{g^\vee-1}\neq0$, $S^{g^\vee}=0$. 
We will do this by first examining the subring $(B\times B)^\fg$ of $\fg$-invariants in $B\times B$,
finding an explicit set of generators for $(B\times B)^\fg$ in terms of Casimirs of $\fg$.
Our proof completely ignores the $\sl(2)$-covariance of $A$ and $\sl(2)$-invariance of $A^\fg$, which otherwise looks like an attractive starting point.

We now focus on the ring $(B\times B)^\fg$, with the purpose of finding its generators.
One obvious generator is $S$.

Since all $r_i$'s are distinct, and since non-self-conjugate $r_i$'s are accompanied by $r_{i'}=r^\star_i$ at the same level in $B$, all $\fg$-invariants are on the ``diagonal'', having bi-degree $(n,n)$. We call this level $n$ in $(B\times B)^\fg$.
One may equip $(B\times B)^\fg$ with a basis which is adapted to the irreducible modules $r_i$.
Given a module $r_i$ at level $n_i$, form the element $X_i=P_iX^{n_i}\in B$ by projection on $r_i$ in $\wedge^{n_i}\fg$.
All possible invariants are linear combinations of
$S_i=(X_i,Y^\star_i)=(X_i,Y^{n_i})=(X^{n_i},Y^\star_i)$, where $Y^\star_i=P^\star_iY^{n_i}$ is projected on $r^\star_i$.
The pairing $(\cdot,\cdot)$ on $B\times B$ is defined with the Cartan--Killing metric. Given elements
$\alpha={1\over n!}X^{a_1\ldots a_n}T_{a_1}\wedge\ldots\wedge T_{a_n}$ and the same for $\beta$ in terms of $Y$,
$(\alpha,\beta)={1\over n!}X^{a_1\ldots a_n}Y_{a_1\ldots a_n}$ by defining 
$(T_{a_1}\wedge\ldots\wedge T_{a_n},T^{b_1}\wedge\ldots\wedge T^{b_n})=n!\delta_{a_1\ldots a_n}^{b_1\ldots b_n}$.
Conversely, to an invariant 
$T=t_{a_1\ldots a_n}^{b_1\ldots b_n}X^{a_1}\ldots X^{a_n}Y_{b_1}\ldots Y_{b_n}$ we may associate an element in $B$:
\begin{align}
X_T=t_{a_1\ldots a_n}^{b_1\ldots b_n}X^{a_1}\ldots X^{a_n}
T_{b_1}\wedge\ldots\wedge T_{b_n}\;.
\label{XSnEq}
\end{align}
For example,
$X_{S^n}=(-1)^{n(n-1)\over2}X^n$.

\begin{lemma}\label{ProductLemma}
The multiplication table of irreducible modules in $B$ and that of $(B\times B)^\fg$ contain the same structure constants.
Let $X_iX_j=\sum_kc_{ij}{}^kX_k$. Then,
$S_iS_j=(-1)^{n_in_j}\sum_kc_{ij}{}^kS_k$.
\end{lemma}

\begin{proof}
\begin{align}
S_iS_j&=(X_i,Y^{n_i})(X_j,Y^{n_j})=(-1)^{n_in_j}(X_iX_j,Y^{n_i+n_j})\nn\\
&=(-1)^{n_in_j}\sum_k(c_{ij}{}^kX_k,Y^{n_k})=(-1)^{n_in_j}\sum_kc_{ij}{}^kS_k\;.
\end{align}
\end{proof}


\begin{cor}\label{SnCor}
Since (with suitable normalisation) $X_S=X$, $S^n\neq0$ for all non-empty levels $n$ in $(B\times B)^\fg$.
$(X_{S^{n_i}},Y_i)\neq0$. 
When expressed in terms of $X_T$'s, where $T$'s are products of generators of $(B\times B)^\fg$, every irreducible component $X_i$ has a non-vanishing coefficient for $X_{S^{n_i}}$.
\end{cor}
The last statement, which will be important for the proof, amounts to the fact that any projection operator $P_i$ on an $r_i$ has a non-vanishing trace with the identity.

With the input from Appendix  \ref{ModulesApp}, we can form the partition function of the ring $(B\times B)^\fg$, which also counts the number of irreducible modules $b_n$ at each level $n$ in $B$, and investigate for patterns. We write the partition function as
\begin{align}
Z_B(t)=\sum_{n=0}^Nb_nt^n\;,
\end{align} 
where $N$ is the highest level in $B$.
One striking property is that the first occurrence of $b_n>1$ is at a level which is one less than the order of the next (after $C_2$) Casimir operator of $\fg$. (That is, if this happens at all. In very simple cases, like $G_2$, the ring may end before that level is reached. Then, $b_n=1$, $n=0,\ldots,N$.)
This observation goes much further. Let us define a ring $C$ with generators at level $n$ only if $n+1\in\Gamma$, where $\Gamma$ is the set of integers $m$ for which $C_m$ is a Casimir of $\fg$
(see Appendix \ref{ModulesApp}). Let there be no relations in $C$. Then $C$ has the partition function\footnote{Note that due to the shift in degree by one unit, $C\not\simeq U(\fg)^\fg$ as a graded ring.}
\begin{align}
Z_C(t)=\prod_{n+1\in\Gamma}(1-t^n)^{-1}=\sum_{n=0}^\infty c_nt^n\;.
\end{align}
We notice, by direct comparison of the two partition functions, that they agree up to level $g^\vee-1$ for all exceptional $\fg$, \ie, 
$b_n=c_n$ for $n\leq g^\vee-1$ and $b_n<c_n$ for $n\geq g^\vee$. 
Analogous comparisons can be made for classical matrix algebras. There, we have verified the same behaviour up to rank $10$, which convinces us that it holds in general, although we have no proof.

For example, we have for $\fg=E_6$,
\begin{align}
Z_B(t)&=1 + t + t^2 + t^3 + 2t^4 + 3t^5 + 3t^6 +4t^7 + 6t^8 + 7t^9 + 8t^{10}+10t^{11} \nn\\ 
      &+7t^{12} +4t^{13} +2t^{14} +2t^{15} +2t^{16}\;,\nn\\
Z_C(t)&=\prod_{n+1\in\{2,5,6,8,9,12\}}(1-t^n)^{-1}\\&=1 + t + t^2 + t^3 + 2t^4 + 3t^5 + 3t^6 +
      4t^7 + 6t^8 + 7t^9 + 8t^{10}+10t^{11} \nn\\ 
      &+13t^{12} +15t^{13} +17t^{14} +21t^{15} +26t^{16}+O(t^{17})\nn\;,
\end{align}
showing agreement up to level $g^\vee-1=11$.

Having observed that the partition functions agree up to level $g^\vee-1$, we will proceed to show that the rings themselves agree to that level. Finding a concrete expression for the generators of $(B\times B)^\fg$ will allow us to state that $(B\times B)^\fg\simeq C/J$ where the ideal is empty below level $g^\vee$.

How can invariants be constructed? Let $\xi$ and $\eta$ be $X$ and $Y$ in some (any) representation $R$ with representation matrices $t_a$, \ie, the matrices $\xi=X^at_a$, $\eta=Y^at_a$. Since $\xi^2=0=\eta^2$ (with matrix multiplication)\footnote{But remember that the ideal generated by $[X,Y]$ is not yet divided out, so $\xi\eta+\eta\xi\neq0$.}, all invariants can be written as traces of alternating $\xi$'s and $\eta$'s:
$\tr(\xi\eta\xi\eta\ldots\xi\eta)$, and products of such expressions. 
It is not yet obvious from such expressions when new independent invariants occur.
The identities $\xi^2=\eta^2=0$ can be used  to rewrite  $\tr((\xi\eta)^n)=\tr(\xi,\eta,[\xi,\eta]^{n-1})$, where $[\xi,\eta]=\xi\eta+\eta\xi$. 
Due to the Jacobi identities (or by just evaluating the matrix products) $[\xi,[\xi,\eta]]=0=[[\xi,\eta],\eta]$. This implies that 
$\tr(\xi,\eta,[\xi,\eta]^{n-1})$ is symmetric in all its $n+1$ entries.
We are looking for invariants in the totally symmetric product of a number of adjoint elements. Such invariants are built from Casimir operators.

Using the Casimir operators of $\fg$, we can construct the generators of $(B\times B)^\fg$ as
\begin{align}
S_{(n)}=C_{n+1}(X,Y,\underbrace{[X,Y],\ldots,[X,Y]}_{n-1})\}\;,\label{SnDef}
\end{align}
where $[X,Y]^a=f_{bc}{}^aX^bY^c$ and $C_m$ are Casimir operators of $\fg$. The expressions should be evaluated in some representation.
Due to symmetry of the Casimir, at most one ``naked'' $X$ and one $Y$ can be among its entries.
There is no way of rewriting $S_{(n)}$ using the identities for invariance. This is since the Jacobi identities 
imply $[X,[X,Y]]=0=[[X,Y],Y]$. 
 It is then also clear from invariance that $C_n([X,Y],\ldots,[X,Y])=0$, so the naked $X$ and $Y$ are needed.
There is an arbitrariness in defining each higher Casimir in that one may add some product of lower ones. One possible canonical choice is that contractions with lower Casimirs vanish. This is however irrelevant in our case; the multi-trace terms added give zero in 
eq. \eqref{SnDef}, since the naked $X$ and $Y$ would be needed in every factor. 
There is of course an arbitrariness in defining the generator at a given level by addition of some constant times products of lower generators. This is unconnected to the choice of representatives for the Casimir invariants, which provide a canonical choice. We also note that if one tries to form an invariant at a level $n$ where there is no Casimir $C_{n+1}$, the single trace $\tr(\xi,\eta,[\xi,\eta]^{n-1})$ is expressible as a sum of products of lower traces, and vanishes identically. The absence of possible further generators of $(B\times B)^\fg$ establishes that $(B\times B)^\fg=C/J$, where the ideal $J$ is empty below level $g^\vee$.

We can thus state:

\begin{prop}\label{CnProp}
The subring of $\fg$-invariants in $B\times B$, $(B\times B)^\fg$, is generated by the set 
\begin{align}
\{S_{(n)}=C_{n+1}(X,Y,\underbrace{[X,Y],\ldots,[X,Y]}_{n-1})\,:\; n+1\in\Gamma\}\;.
\end{align}
The first relations between products of $S_{(n)}$'s appear at level $g^\vee$.
\end{prop}

When we form the ring $A^\fg$, it will be $(B\times B)^\fg/I$ for some ideal $I$. We see that the generators $S_{(n)}$, $n>1$ belong to $I$. Conversely, any element in $I$ will consist of terms, each containing a trace with at least one $[\xi,\eta]$, which again is expressible using $S_{(n)}$, $n>1$. So, $I$ is generated by all generators of $(B\times B)^\fg$ except $S$.

\begin{cor}
$A^\fg$ is generated by
$S$, and $S^{g^\vee-1}\neq0$ in $A^\fg$. 
\end{cor}

However, we still need to check that the relations in $(B\times B)^\fg$ involve $S^{g^\vee}$ in a way that implies that it vanishes in $A^\fg$. 
This is the final step in the proof of the CDSW conjecture for exceptional $\fg$.


\begin{prop}\label{ShProp}
$S^{g^\vee}=0$ in $A^\fg$.
\end{prop}

\begin{proof}
We will use Lemma \ref{ProductLemma} to show that $S^{g^\vee}=0$ for the exceptional algebras
$G_2$, $F_4$, $E_6$, $E_7$ and $E_8$. 
In each of the five cases, remarkably enough, there is a single irreducible module $r_0$ in $B$ at level $g^\vee-1$ which does not propagate to level $g^\vee$, \ie, $XX_0=0$. 
These ``local endpoints" of $B$ are indicated in red and underlined in the tables of Appendix \ref{ModulesApp}.
By Lemma \ref{ProductLemma}, $SS_0=0$. For $G_2$ there is a single module at level 3 and none at level 4, so the statement $S^4=0$ holds already in $(B\times B)^\fg$. In all other cases, $S_0$, suitably normalised, can be expressed as $S^{g^\vee-1}+S'$, where $S'$ is some linear combination of invariants at level $g^\vee-1$, all of which contain at least one $S_{(n)}$, $n>1$. This follows from Corollary \ref{SnCor}. Thus, $0=SS_0=S^{g^\vee}+SS'$ in $(B\times B)^\fg$. In $A^\fg$, $S'=0$, so $S^{g^\vee}=0$.
\end{proof}

This completes the proof of the different parts of the CDSW conjecture, and we can state

\begin{thm}
The CDSW conjecture holds for exceptional Lie algebras.
\end{thm}

\vspace{2\parskip}
\noindent{\large\bf 3. Final remarks}

\noindent To conclude: The explicit form of the generators enabled us to identify the generators of $(B\times B)^\fg$ with Casimir operators. This, and thus the generation of $A^\fg$ by $S$ (shown by Kumar in \cite{Kumar2004}), holds for any simple $\fg$. The absence of relations below level $g^\vee$ is obtained by a counting argument for the exceptional Lie algebras. The observation that $S^{g^\vee-1}\neq0$ in $A^\fg$ then follows. Our proof of $S^{g^\vee}=0$ for the exceptional algebras uses more detailed information of the representation content of the rings.

We have focussed on the exceptional Lie algebras. All statements about $B$, $B\times B$ and $(B\times B)^\fg$ should hold equally for the classical matrix algebras. In particular, Lemma \ref{ProductLemma} holds, and Prop. \ref{CnProp} holds with the same proof, 
if universal countings of $b_n$ up to level $g^\vee$ are established.
We conjecture that the proof of Prop. \ref{ShProp} holds for classical matrix algebras of types $A$, $B$ and $D$ with the following local endpoints $r_0$ at level $g^\vee-1$:
\begin{center}
\begin{tabular}{|c|c|c|}
\hline
Algebra&$g^\vee$&$r_0$\\
\noalign{\global\arrayrulewidth=.8pt}
\hline
\noalign{\global\arrayrulewidth=.4pt}
$A_1$&2&(2)\\
\hline
$A_r$, $r\geq2$&$r+1$&$(r+1,0\ldots0)\oplus(0\ldots0,r+1)$\\
\hline
$B_3$&5&(104)\\
\hline
$D_4$&$6$&$\DWeight{20}22$\\
\hline
$\so(d)=\left\{\begin{matrix}D_r\,,\;d=2r\phantom{,+1}\\B_r\,,\;d=2r+1\end{matrix}\right.\quad d\geq9$
&$d-2$&$(d-6,020\ldots0)$\\
\hline
\end{tabular}
\end{center}
This is experimentally observed for ranks up to $10$, but it should be possible to give a universal proof.
For the algebras $C_r$, $g^\vee=r+1$, and $b_{r+1}=c_{r+1}-1$ (the difference by 1 is also experimental observation), but there is no local endpoint of $B$ at level $r$.
Still, there will obviously be one relation between the elements of $C$ at level $r+1$, and one needs to find an alternative method for proving that the linear combination set to $0$ has a non-vanishing coefficient for  $S^{r+1}$. This has effectively been done in ref. \cite{Witten:2003ye}.

\vspace{12pt}

\noindent\underline{\it Acknowledgments:} We are grateful to Pavel Etingof for pointing out some less precise statements and for discussions during the preparation of version 2 of the paper.

\newpage
\appendix

\section{The content of $\fg$-modules in $B$ for exceptional $\fg$\label{ModulesApp}}

The irreducible modules are labelled by their highest weights, which are expressed in terms of their coefficients in the basis of fundamental weights. The ``ordering'' is that of the respective Dynkin diagram. When occasionally a two-digit coefficient appears, it is placed within parentheses. The non-propagating modules $r_0$ at level $g^\vee-1$ are indicated in red and underlined.

\subsection{$G_2$}

\pgfkeys{/Dynkin diagram, edge length=1cm,
fold radius=.6cm, 
root-radius=.14cm,
indefinite edge/.style={
draw=black, fill=white, dotted,
thin}}

\begin{center}
\dynkin{G}{oo}
\end{center}

$g^\vee=4$, $\Gamma=\{2,6\}$.

\begin{center}
\begin{tabular}{|c|c|c|}
\hline
level&irreducible modules&$b_n$\\
\noalign{\global\arrayrulewidth=.8pt}
\hline
\noalign{\global\arrayrulewidth=.4pt}
0&\GTwo00&1\\
\hline
1&\GTwo01&1\\
\hline
2&\GTwo30&1\\
\hline
3&\underline{\Red{\GTwo40}}&1\\[2pt]
\noalign{\arrayrulewidth=.4pt}
\hline
\end{tabular}
\end{center}

(There is a misprint for the level 3 module in ref.  \cite{EtingofKac}.)

\subsection{$F_4$}

\begin{center}
\dynkin{F}{oooo}
\end{center}

$g^\vee=9$, $\Gamma=\{2,6,8,12\}$.

\begin{center}
\begin{tabular}{|c|c|c|}
\hline
level&irreducible modules&$b_n$\\
\noalign{\global\arrayrulewidth=.8pt}
\hline
\noalign{\global\arrayrulewidth=.4pt}
0&(0000)&1\\
\hline
1&(1000)&1\\
\hline
2&(0100)&1\\
\hline
3&(0020)&1\\
\hline
4&(0021)&1\\
\hline
5&(0030)\;(0103)&2\\
\hline
6&(0112)\;(1005)&2\\
\hline
7&(0007)\;(0202)\;(1014)&3\\
\hline
8&\underline{\Red{(0300)}}\;(0016)\;(1104)&3\\[2pt]
\noalign{\global\arrayrulewidth=.8pt}
\hline
\noalign{\global\arrayrulewidth=.4pt}
9&(0106)&1\\
\hline
\end{tabular}
\end{center}

\newpage

\subsection{$E_6$}

\begin{center}
\dynkin EI
\end{center}

$g^\vee=12$, $\Gamma=\{2,5,6,8,9,12\}$.

\begin{center}
\begin{tabular}{|c|c|c|}
\hline
level&irreducible modules&$b_n$\\
\noalign{\global\arrayrulewidth=.8pt}
\hline
\noalign{\global\arrayrulewidth=.4pt}
0&\ESix000000&1\\
\hline
1&\ESix010000&1\\
\hline
2&\ESix000100&1\\
\hline
3&\ESix001010&1\\
\hline
4&\ESix002001 \ESix100020&2\\
\hline
5&\ESix000030 \ESix003000 \ESix101011&3\\
\hline
6&\ESix001021 \ESix102010 \ESix200102&3\\
\hline
7&\ESix002020 \ESix100112 \ESix201101 
     \ESix310003&4\\
\hline
8&\ESix000203 \ESix300200 \ESix101111 \ESix210013 
      \ESix311002 \ESix400004&6\\
\hline
9&\ESix001202 \ESix200210 \ESix110104 \ESix410101
     \ESix211012 \ESix300014 \ESix401003 
     &7\\
\hline
10& \ESix021005 \ESix520010 \ESix100301 \ESix310111 \ESix111103 
     \ESix200105 \ESix500102 \ESix301013 
     &8\\
\hline
11&\underline{\Red{${\EWeight{00}{4}{00}{0}}$}}
\ESix022004 \ESix420020 \ESix030006 \ESix630000
     \ESix111006 \ESix610011 \ESix201104 \ESix400112 \ESix210202 
     &10\\[4pt]
\noalign{\global\arrayrulewidth=.8pt}
\hline
\noalign{\global\arrayrulewidth=.4pt}
12&\ESix002007 \ESix700020 \ESix112005 \ESix510021 \ESix120007 
     \ESix720001 \ESix300203 
     &7\\
\hline
13&\ESix003006 \ESix600030 \ESix011008 
     \ESix810010&4\\
\hline
14&\ESix000109 \ESix900100&2\\
\hline
15& $\EWeight{00}0{1(10)}0$ $\EWeight{(10)1}0{00}0$&2\\
\hline
16&$\EWeight{(12)0}0{00}0$ $\EWeight{00}0{0(12)}0$&2\\
\hline
\end{tabular}
\end{center}

\subsection{$E_7$}

\begin{center}
\dynkin[backwards=true]E{ooooooo}
\end{center}

$g^\vee=18$, $\Gamma=\{2,6,8,10,12,14,18\}$.

\begin{center}
\begin{tabular}{|c|c|c|}
\hline
level&irreducible modules&$b_n$\\
\noalign{\global\arrayrulewidth=.8pt}
\hline
\noalign{\global\arrayrulewidth=.4pt}
0&\ESeven0000000&1\\
\hline
1&\ESeven1000000&1\\
\hline
2&\ESeven0010000&1\\
\hline
3&\ESeven0001000&1\\
\hline
4&\ESeven0100100&1\\
\hline
5&\ESeven0000200 \ESeven0200010&2\\
\hline
6&\ESeven0100110 \ESeven0300001&2\\
\hline
7&\ESeven0001020 \ESeven0200101 \ESeven0400000&3\\
\hline
8&\ESeven0010030 \ESeven0101011 \ESeven0300100&3\\
\hline
9&\ESeven0002002 \ESeven0110021 \ESeven0201010 
     \ESeven1000040&4\\
\hline
10&\ESeven0000050 \ESeven0011012 \ESeven0102001 
     \ESeven0210020 \ESeven1100031&5\\
\hline
11&\ESeven0003000 \ESeven0020103 \ESeven0100041 
     \ESeven0111011 \ESeven1001022 \ESeven1200030&6\\
\hline
12&\ESeven0001032 \ESeven0012010 \ESeven0030004 
     \ESeven0120102 \ESeven0200040 \ESeven1010113 
     \ESeven1101021&7\\
\hline
13&\ESeven0010123 \ESeven0021101 \ESeven0101031 
     \ESeven0130003 \ESeven1002020 \ESeven1020014 
     \ESeven1110112 \ESeven2000204&8\\
\hline
14&\makecell{\ESeven0002030 \ESeven0020024 \ESeven0030200 
     \ESeven0031002 \ESeven0110122 \ESeven1000214 
     \ESeven1011111 \ESeven1120013 \\ \ESeven2010105 
     \ESeven2100203}&10\\
\hline
15&\makecell{\ESeven0000305 \ESeven0011121 \ESeven0040101 
     \ESeven0120023 \ESeven1010115 \ESeven1020210 
     \ESeven1021012 \ESeven1100213 \\ \ESeven2001202 
     \ESeven2110104 \ESeven3001006}&11\\
\hline
16&\makecell{\ESeven0010206 \ESeven0020220 \ESeven0021022 
     \ESeven0050010 \ESeven0100304 \ESeven1001212 
     \ESeven1030111 \ESeven1110114 \\ \ESeven2001016 
     \ESeven2010301 \ESeven2011103 \ESeven3101005 
     \ESeven4100007}&13\\
\hline
17&\makecell{\underline{\Red{$\EWeight{000}{0}{60}{0}$}}
\ESeven0001303 \ESeven0030121 
     \ESeven0110205 \ESeven1001107 \ESeven1010311 
     \ESeven1011113 \ESeven1040020 \\ \ESeven2020202 
     \ESeven2101015 \ESeven3000400 \ESeven3002004 
     \ESeven3100017 \ESeven4200006 \ESeven5000008}&15\\
\noalign{\global\arrayrulewidth=.8pt}
\hline
\noalign{\global\arrayrulewidth=.4pt}
18&\makecell{\ESeven0002008 \ESeven0010402 \ESeven0011204 
     \ESeven0040030 \ESeven1020212 \ESeven1101106 
     \ESeven2000410 \ESeven2002014 \\ \ESeven2100108 
     \ESeven3200016 \ESeven4000018}&11\\
\hline
19&\ESeven0020303 \ESeven0102007 \ESeven1000501 
     \ESeven1002105 \ESeven1101009 \ESeven2200107 
     \ESeven3000109&7\\
\hline
20&\ESeven0000600 \ESeven0003006 \ESeven021000{(10)} 
     \ESeven1201008 \ESeven200100{(10)}&5\\
\hline
21&\ESeven030000{(11)} \ESeven0310009 \ESeven111000{(11)}&3\\
\hline
22&\ESeven002000{(12)} \ESeven040000{(10)} \ESeven120000{(12)}&3\\
\hline
23&\ESeven011000{(13)}&1\\
\hline
24&\ESeven000100{(14)}&1\\
\hline
25&\ESeven000010{(15)}&1\\
\hline
26&\ESeven000001{(16)}&1\\
\hline
27&\ESeven000000{(18)}&1\\
\hline
\end{tabular}
\end{center}

\subsection{$E_8$}

\begin{center}
\dynkin[backwards=true]E{oooooooo}
\end{center}

$g^\vee=30$, $\Gamma=\{2,8,12,14,18,20,24,30\}$.

\begin{center}
\begin{tabular}{|c|c|c|}
\hline
level&irreducible modules&$b_n$\\
\noalign{\global\arrayrulewidth=.8pt}
\hline
\noalign{\global\arrayrulewidth=.4pt}
0&\EEight00000000&1\\
\hline
1&\EEight00000001&1\\
\hline
2&\EEight00000010&1\\
\hline
3& \EEight00000100&1\\
\hline
4& \EEight00001000&1\\
\hline
5&\EEight00010000&1\\
\hline
6&\EEight01100000&1\\
\hline
7&\EEight00200000 \EEight12000000&2\\
\hline
8&\EEight03000000 \EEight11100000&2\\
\hline
9&\EEight02100000 \EEight20010000&2\\
\hline
10&\EEight11010000 \EEight30001000&2\\
\hline
11&\EEight00020000 \EEight21001000 \EEight40000100&3\\
\hline
12&\EEight10011000 \EEight31000100 \EEight50000010&3\\
\hline
13&\EEight00102000 \EEight20010100 \EEight41000010 
     \EEight60000001&4\\
\hline
14&  \EEight00003000 \EEight10101100 \EEight30010010 
     \EEight51000001 \EEight70000000&5\\
\hline
15& \EEight00200200 \EEight10002100 \EEight20101010 
     \EEight40010001 \EEight61000000&5\\
\hline
16&\EEight00101200 \EEight10200110 \EEight20002010 
     \EEight30101001 \EEight50010000&5\\
\hline
17&\EEight00010300 \EEight00300020 \EEight10101110 
     \EEight20200101 \EEight30002001 \EEight40101000&6\\
\hline
18&\EEight00201020 \EEight01000400 \EEight10010210 
     \EEight10300011 \EEight20101101 \EEight30200100 
     \EEight40002000&7\\
\hline
19&\EEight00000500 \EEight00110120 \EEight00400002 
     \EEight10201011 \EEight11000310 \EEight20010201 
     \EEight20300010 \EEight30101100&8\\
\hline
20&\makecell{\EEight00020030 \EEight00301002 \EEight01100220 
     \EEight10000410 \EEight10110111 \EEight10400001 
     \EEight20201010 \EEight21000301\\ \EEight30010200}&9
\\
\hline
\end{tabular}
\end{center}

\newpage
\begin{center}
$E_8$, cont'd.
\end{center}

\begin{center}
\begin{tabular}{|c|c|c|}
\hline
level&irreducible modules&$b_n$\\
\hline
21&\makecell{\EEight00100320 \EEight00210102 \EEight00500000 
     \EEight01010130 \EEight10020021 \EEight10301001 
     \EEight11100211 \EEight20000401 \\ \EEight20110110 
     \EEight31000300}&10\\
\hline
22&\makecell{\EEight00010230 \EEight00120012 \EEight00401000 
     \EEight01200202 \EEight02001040 \EEight10100311 
     \EEight10210101 \EEight11010121\\ \EEight20020020 
     \EEight21100210 \EEight30000400}&11\\
\hline
23&\makecell{\EEight00030003 \EEight00200302 \EEight00310100 
     \EEight01001140 \EEight01110112 \EEight03000050 
     \EEight10010221 \EEight10120011\\ \EEight11200201 
     \EEight12001031 \EEight20100310 \EEight21010120}&12\\
\hline
24&\makecell{\EEight00002050 \EEight00110212 \EEight00220010 
     \EEight01020103 \EEight01300200 \EEight02000150 
     \EEight02101022 \EEight10030002\\ \EEight10200301 
     \EEight11001131 \EEight11110111 \EEight13000041 
     \EEight20010220 \EEight22001030}&14\\
\hline
25&\makecell{\EEight00020203 \EEight00130001 \EEight00300300 
     \EEight01001060 \EEight01101122 \EEight01210110 
     \EEight02011013 \EEight03100032\\ \EEight10002041 
     \EEight10110211 \EEight11020102 \EEight12000141 
     \EEight12101021 \EEight21001130 \EEight23000040}&15\\
\hline
26&\makecell{\EEight00010070 \EEight00040000 \EEight00102032 
     \EEight00210210 \EEight01011113 \EEight01120101 
     \EEight02100132 \EEight02201020\\ \EEight03002004 
     \EEight03010023 \EEight10020202 \EEight11001051 
     \EEight11101121 \EEight12011012 \EEight13100031 
     \EEight20002040\\ \EEight22000140}&17\\
\hline
27&\makecell{\EEight00012023 \EEight00100080 \EEight00120201 
     \EEight01030100 \EEight01101042 \EEight01201120 
     \EEight02002104 \EEight02010123\\ \EEight02111011 
     \EEight03200030 \EEight04001014 \EEight10010061 
     \EEight10102031 \EEight11011112 \EEight12100131 
     \EEight13002003\\ \EEight13010022 \EEight21001050}&18\\
\hline
28&\makecell{\EEight00030200 \EEight00110052 \EEight00202030 
     \EEight01003014 \EEight01011033 \EEight01111111 
     \EEight02021010 \EEight02200130\\ \EEight03001114 
     \EEight03102002 \EEight03110021 \EEight05000105 
     \EEight10000090 \EEight10012022 \EEight10100071 
     \EEight11101041\\ \EEight12002103 \EEight12010122 
     \EEight14001013 \EEight20010060}&20\\
\hline
29&\makecell{\underline{\Red{$\EWeight{0(10)00}{0}{00}{0}$}}
\EEight00004005 \EEight00020043 
     \EEight00112021 \EEight00200062 \EEight01021110 
     \EEight01201040 \EEight02002024\\ \EEight02102102 
     \EEight02110121 \EEight03012001 \EEight03020020 
     \EEight04000205 \EEight04101012 \EEight06000006 
     \EEight10110051\\ \EEight11003013 \EEight11011032 
     \EEight13001113 \EEight15000104 \EEight20000081 
     \EEight20100070}&22\\
\noalign{\global\arrayrulewidth=.8pt}
\hline
\noalign{\global\arrayrulewidth=.4pt}
30&\makecell{\EEight00022020 \EEight00210050 \EEight01103012 
     \EEight01111031 \EEight02012101 \EEight02020120 
     \EEight03101112 \EEight04003000\\ \EEight04011011 
     \EEight05100103 \EEight10004004 \EEight10020042 
     \EEight10200061 \EEight12002023 \EEight14000204 
     \EEight16000005\\ \EEight30000080}&17\\
\hline
31&\makecell{\EEight00104003 \EEight00120041 \EEight00300060 
     \EEight01013011 \EEight01021030 \EEight02102022 
     \EEight03003100 \EEight03011111\\ \EEight04100203 
     \EEight05002010 \EEight05010102 \EEight06100004}&12\\
\hline
32&\EEight00014002 \EEight00030040 \EEight02004010 
     \EEight02012021 \EEight04002110 \EEight04010202 
     \EEight06001101 \EEight06010003&8\\
\hline
33&\EEight01005001 \EEight03003020 \EEight05001201 
     \EEight07000200 \EEight07001002&5\\
\hline
34&\EEight00006000 \EEight06000300 \EEight08000101&3\\
\hline
35&\EEight09000010&1\\
\hline
36&$\EWeight{0000}{\hspace{-.5pt}0}{\!\!\!\hspace{-1.5pt}00}{\hspace{-5.5pt}(10)}$&1\\
\hline
\end{tabular}
\end{center}

\bibliographystyle{utphysmod2}



\begin{thebibliography}{1}

\bibitem{Cachazo:2002ry}
F.~Cachazo, M.~R. Douglas, N.~Seiberg and E.~Witten,  {\em {Chiral rings and
  anomalies in supersymmetric gauge theory}}, JHEP {\bf 12}, 071 (2002)
  [\href{http://www.arXiv.org/abs/hep-th/0211170}{{\tt hep-th/0211170}}].

\bibitem{Witten:2003ye}
E.~Witten,  {\em {Chiral ring of Sp(N) and SO(N) supersymmetric gauge theory in
  four-dimensions}}, \href{http://www.arXiv.org/abs/hep-th/0302194}{{\tt
  hep-th/0302194}}.

\bibitem{EtingofKac}
P.~Etingof and V.~Kac,  {\em {On the Cachazo--Douglas--Seiberg--Witten
  conjecture for simple Lie algebras}},
  \href{http://www.arXiv.org/abs/math/0305175}{{\tt math/0305175}}.

\bibitem{Kumar2004}
S.~Kumar,  {\em {On the Cachazo--Douglas--Seiberg--Witten conjecture for simple
  Lie algebras}}, J. Am. Math. Soc. {\bf 21}, 797 (2008)
  [\href{http://www.arXiv.org/abs/math/0412201}{{\tt math/0412201}}].

\bibitem{Kostant2000}
B.~Kostant,  {\em {On $\wedge\fg$ for a semisimple Lie algebra $\fg$, as an
  equivariant module over the symmetric algebra $S(\fg)$}}, Adv. Studies in
  Pure Math. {\bf 26}, 129 (2000).

\bibitem{PetersonUnpub}
D.~Peterson,  unpublished.

\bibitem{Kostant1998}
B.~Kostant,  {\em {The set of Abelian ideals of a Borel subalgebra, Cartan
  decompositions, and discrete series representations}}, Int. Math. Res.
  Notices {\bf 1998}, 225 (1998).

\bibitem{LiE}
M.~van Leeuwen, A.~Cohen and B.~Lisser,  {LiE}.
  \url{http://www-math.univ-poitiers.fr/~maavl/LiE/}.

\end{thebibliography}

\providecommand{\href}[2]{#2}\begingroup\raggedright\endgroup

\end{document}